\documentclass[runningheads]{llncs}

\usepackage[T1]{fontenc}
\usepackage{amssymb}
\usepackage{subcaption}
\usepackage{tikz}
\usetikzlibrary{automata, arrows.meta, positioning, calc}

\tikzstyle{accepting}=[double distance=2pt, outer sep=1pt+\pgflinewidth]
\tikzstyle{accepting}=[path picture={%
  \draw let
    \p1 = (path picture bounding box.east),
    \p2 = (path picture bounding box.center)
    in
      (\p2) circle (\x1 - \x2 - 2pt);
  }]

\begin{document}

\title{Universal First-Order Quantification over
  Automata\thanks{Research reported in this paper was supported in part
    by an Amazon Research Award, Fall 2022 CFP. Any opinions, findings, and
    conclusions or recommendations expressed in this material are
    those of the authors and do not reflect the views of Amazon.}}

\author{Bernard Boigelot\orcidID{0009-0009-4721-3824} \and
Pascal Fontaine\orcidID{0000-0003-4700-6031} \and
Baptiste Vergain\orcidID{0009-0003-5545-4579}}

\authorrunning{B. Boigelot, P. Fontaine, and B. Vergain}

\institute{Montefiore Institute, B28, University of Liège, Belgium}

\maketitle

\vspace*{-1ex}
\begin{abstract}
  Deciding formulas mixing arithmetic and uninterpreted predicates is
  of practical interest, notably for applications in
  verification. Some decision procedures consist in building by
  structural induction an automaton that recognizes the set of models
  of the formula under analysis, and then testing whether this
  automaton accepts a non-empty language. A drawback is that universal
  quantification is usually handled by a reduction to existential
  quantification and complementation. For logical formalisms in which
  models are encoded as infinite words, this hinders the practical use
  of this method due to the difficulty of complementing infinite-word
  automata. The contribution of this paper is to introduce an
  algorithm for directly computing the effect of universal first-order
  quantifiers on automata recognizing sets of models, for formulas
  involving natural numbers encoded in unary notation.  This makes it
  possible to apply the automata-based approach to obtain
  implementable decision procedures for various arithmetic theories.
\end{abstract}


\section{Introduction}

Automated reasoning with arithmetic theories is of primary importance,
notably for verification, where Satisfiability Modulo Theories (SMT)
solvers are regularly used to discharge proof obligations.  It is well
known however that mixing arithmetic with uninterpreted symbols
quickly leads to undecidable languages.  For instance, extending
Presburger arithmetic, i.e., the first-order additive theory of
integer numbers, with just one uninterpreted unary predicate makes
it undecidable~\cite{Downey72,Halpern91,Speranski13}.  There exist
decidable fragments mixing arithmetic and uninterpreted symbols that
are expressive enough to be interesting, for instance, the monadic
second-order theory of $\mathbb{N}$ under one successor (S1S).

The decidability of S1S has been established thanks to the concept of
infinite-word automaton~\cite{Buchi62}.  In order to decide whether a
formula $\varphi$ is satisfiable, the approach consists in building an
automaton that recognizes the set of its models, encoded in a suitable
way, and then checking that this automaton accepts a non-empty
language.  Such an automaton has one separate input tape for each
first-order and second-order free variable of $\varphi$. It is
constructed by starting from elementary automata representing the
atoms of $\varphi$, and then translating the effect of Boolean
connectives and quantifiers into corresponding operations over automata.
For instance, applying an existential quantifier simply amounts to
removing from the automaton the input tape associated to the
quantified variable. Universal quantification reduces to existential
quantification thanks to the equivalence $\forall x\, \varphi
\,\equiv\, \neg\exists x \,\neg\varphi$.

Even though this approach has originally been introduced as a purely
theoretical tool, it is applied in practice to obtain usable decision
procedures for various logics.  In particular, the tool
MONA~\cite{Klarlund97} uses this method to decide a restricted version
of S1S, and tools such as LASH~\cite{BL04} and Shasta~\cite{SKR98} use
a similar technique to decide Presburger arithmetic. The former tool
also generalizes this result by providing an implemented decision
procedure for the first-order additive theory of mixed integer and
real variables~\cite{BJW05}.

A major issue in practice is that the elimination of universal
quantifiers relies on complementation, which is an operation that is
not easily implemented for infinite-word
automata~\cite{Safra88,Vardi07}.  Actual implementations of
automata-based decision procedures elude this problem by restricting
the language of interest or the class of automata that need to be
manipulated. For instance, the tool MONA only handles \textit{Weak}\,
S1S (WS1S) which is, schematically, a restriction of S1S to finite
subsets of natural numbers~\cite{Buchi60}. The tool LASH handles the
mixed integer and real additive arithmetic by working with
\textit{weak deterministic automata}, which are a restricted form of
infinite-word automata admitting an easy complementation
algorithm~\cite{BJW05}.

The contribution of this paper is to introduce a direct algorithm for
computing the effect of universal first-order quantification over
infinite-word automata. This is an essential step towards practical
decision procedures for more expressive fragments mixing arithmetic
with uninterpreted symbols. The considered automata are those that
recognize models of formulas over natural numbers encoded in
unary notation. This algorithm does not rely on complementation, and
can be implemented straightforwardly on unrestricted infinite-word
automata.  As an example of its potential applications, this algorithm
leads to a practically implementable decision procedure for the
first-order theory of natural numbers with the order relation and
uninterpreted unary predicates. It also paves the way to a decision
procedure for SMT solvers for the UFIDL (Uninterpreted
Functions and Integer Difference Logic) logic with
only unary predicates.

\section{Basic notions}
\label{sec:basics}

\subsection{Logic}

We address the problem of deciding satisfiability for formulas
expressed in first-order structures of the form $(\mathbb{N}, R_1,
R_2, \ldots)$, where $\mathbb{N}$ is the \textit{domain} of natural
numbers, and $R_1$, $R_2$, \ldots\ are (interpreted)
\textit{relations} over tuples of values in $\mathbb{N}$. More precisely, each
$R_i$ is defined as a relation $R_i \subseteq \mathbb{N}^{\alpha_i}$ for some
$\alpha_i > 0$ called the \textit{arity} of $R_i$.

The formulas in such a structure involve \textit{first-order
variables} $x_1$, $x_2$, \ldots, and \textit{second-order variables}
$X_1$, $X_2$, \ldots\ Formulas are recursively defined as
\begin{itemize}
\item
$\top$, $\bot$,
$x_i = x_j$, $X_i = X_j$,
$X_i(x_j)$ or $R_i(x_{j_1}, \ldots, x_{j_{\alpha_i}})$, where $i, j, j_1, j_2,
\ldots \in {\mathbb{N}}_{>0}$ (\textit{atomic formulas}),
\item
$\varphi_1 \wedge \varphi_2$,
 $\varphi_1 \vee \varphi_2$ or $\neg \varphi$, where
$\varphi_1$, $\varphi_2$ and $\varphi$ are formulas, or
\item
$\exists x_i\, \varphi$ or $\forall x_i\, \varphi$,
where $\varphi$ is a formula.
\end{itemize}
We write $\varphi(x_1, \ldots, x_k, X_1, \ldots, X_{\ell})$ to express
that $x_1, \ldots, x_k, X_1, \ldots, X_{\ell}$ are the free variables
of $\varphi$, i.e., that $\varphi$ does not involve other unquantified
variables.

An \textit{interpretation} $I$ for a formula $\varphi(x_1, \ldots,
x_k, X_1, \ldots, X_{\ell})$ is an assignment of values $I(x_i) \in
\mathbb{N}$ for all $i \in [1, k]$ and $I(X_j) \subseteq \mathbb{N}$
for all $j \in [1, \ell]$ to its free variables. An interpretation $I$ that
makes $\varphi$ true, which is denoted by $I \models \varphi$, is  called
a \textit{model} of $\varphi$.

The semantics is defined in the usual way. One has
\begin{itemize}
\item
$I \models \top$ and $I \not\models \bot$ for every $I$.
\item
$I \models x_i = x_j$ and $I \models X_i = X_j$ iff (respectively)
$I(x_i) = I(x_j)$ and $I(X_i) = I(X_j)$.
\item
$I \models X_i(x_j)$ iff $I(x_j) \in I(X_i)$.
\item
$I \models R_i(x_{j_1}, \ldots, x_{j_{\alpha_i}})$ iff
$(I(x_{j_1}), \ldots, I(x_{j_{\alpha_i}})) \in R_i$.
\item
$I \models \varphi_1 \wedge \varphi_2$, $I \models \varphi_1 \vee \varphi_2$
and $I \models \neg\varphi$ iff (respectively) $(I \models \varphi_1) \wedge
(I \models \varphi_2)$, $(I \models \varphi_1) \vee
(I \models \varphi_2)$, and $I \not\models \varphi$.
\item
$I \models \exists x_i\, \varphi(x_1, \ldots, x_k, X_1, \ldots, X_{\ell})$
iff there exists $n \in \mathbb{N}$ such that
$I[x_i = n] \models \varphi(x_1, \ldots,$ $x_k, X_1, \ldots, X_{\ell})$.
\item
$I \models \forall x_i\, \varphi(x_1, \ldots, x_k, X_1, \ldots, X_{\ell})$
iff for every $n \in \mathbb{N}$, one has
$I[x_i = n] \models \varphi(x_1, \ldots, x_k,$ $X_1, \ldots, X_{\ell})$.
\end{itemize}
In the two last rules, the notation $I[x_i=n]$, where $n \in \mathbb{N}$,
stands for the extension of the interpretation $I$ to one additional
first-order variable
$x_i$ that takes the value $n$, i.e., the interpretation such that
$I[x_i=n](x_j) = I(x_j)$ for all $j \in [1, k]$ such that $j \neq i$,
$I[x_i=n](x_i) = n$, and
$I[x_i=n](X_j) = I(X_j)$ for all $j \in [1, \ell]$.

A formula is said to be \textit{satisfiable} if it admits a model.

\subsection{Automata}
\label{sec:automata}

A finite-word or infinite-word automaton is a tuple ${\cal A} =
(\Sigma, Q, \Delta, Q_0, F)$ where $\Sigma$ is a finite
\textit{alphabet}, $Q$ is a finite set of \textit{states}, $\Delta
\subseteq Q \times (\Sigma \cup \{ \varepsilon \}) \times Q$ is a
\textit{transition relation}, $Q_0 \subseteq Q$ is a set of
\textit{initial states}, and $F \subseteq Q$ is a set of
\textit{accepting states}.

A \textit{path} of $\cal A$ from $q_0$ to $q_m$, with $q_0, q_m \in Q$
and $m \geq 0$, is a finite sequence $\pi = (q_0, a_0, q_1); (q_1,
a_1, q_2); \ldots; (q_{m-1}, a_{m-1}, q_m)$ of transitions from
$\Delta$. The finite word $w \in \Sigma^*$ \textit{read} by $\pi$ is
$w = a_0a_1 \ldots a_{m-1}$; the existence of such a path is denoted
by $q_0 \stackrel{w}{\rightarrow} q_m$. A \textit{cycle} is a
non-empty path from a state to itself.
If ${\cal A}$ is a finite-word automaton, then a path $\pi$ from
$q_0$ to $q_m$ is \textit{accepting} if $q_m \in F$. A word $w \in
\Sigma^*$ is \textit{accepted} from the state $q_0$ if there exists
an accepting path originating from $q_0$ that reads $w$.

For infinite-word automata, we use a Büchi acceptance condition for
the sake of simplicity, but the results of this paper
straightforwardly generalize to other types of infinite-word automata.
If $\cal A$ is an infinite-word automaton, then a
\textit{run} of ${\cal A}$ from a
state $q_0 \in Q$ is an infinite sequence $\sigma = (q_0, a_0, q_1);
(q_1, a_1, q_2); \ldots$ of transitions from $\Delta$. This run reads
the infinite word $w = a_0 a_1 \ldots \in \Sigma^{\omega}$. The run
$\sigma$ is \textit{accepting} if the set $\mathit{inf}(\sigma)$
formed by the states $q_i$ that occur infinitely many times in
$\sigma$ is such that $\mathit{inf}(\sigma) \cap F \neq \emptyset$,
i.e., there exists a state in $F$ that is visited infinitely often by
$\sigma$. A word $w \in \Sigma^{\omega}$ is \textit{accepted} from the
state $q_0 \in Q$ if there exists an accepting run from $q_0$ that
reads $w$.

For both finite-word and infinite-word automata, a word $w$ is
\textit{accepted} by $\cal A$ if it is accepted from an initial state
$q_0 \in Q_0$. The set of all words accepted from a state $q \in Q$
(resp. by $\cal A)$ forms the \textit{language} $L({\cal A}, q)$
accepted from $q$ (resp. $L({\cal A})$ accepted by $\cal A$).  An
automaton accepting $L({\cal A}, q)$ can be derived from ${\cal A}$ by
setting $Q_0$ equal to $\{ q \}$. The language of finite-words $w$
read by paths from $q_1$ to $q_2$, with $q_1, q_2 \in Q$, is denoted by
$L({\cal A}, q_1, q_2)$; a finite-word automaton accepting this
language can be obtained from $\cal A$ by setting $Q_0$ equal to $\{
q_1 \}$ and $F$ equal to $\{ q_2 \}$. A language is said to be
\textit{regular} (resp.  \textit{$\omega$-regular}) if it can be
accepted by a finite-word (resp. an infinite-word) automaton.

\section{Deciding Satisfiability}
\label{sec:decision}

\subsection{Encoding Interpretations}

In order to decide whether a formula $\varphi(x_1, \ldots, x_k, X_1,
\ldots, X_{\ell})$ is satisfiable, Büchi introduced the idea of
building an automaton that accepts the set of all models of $\varphi$,
encoded in a suitable way, and then checking whether it accepts a
non-empty language~\cite{Buchi60,Buchi62}.

A simple encoding scheme consists in representing the value of
first-order variables $x_i$ in \textit{unary notation}: A number $n
\in \mathbb{N}$ is encoded by the infinite word $0^n 1 0^{\omega}$
over the alphabet $\{ 0, 1 \}$, i.e., by a word in which the symbol
$1$ occurs only once, at the position given by $n$. This leads to a
compatible encoding scheme for the values of second-order variables
$X_j$: a predicate $P \subseteq \mathbb{N}$ is encoded by the infinite
word $a_0 a_1 a_2 \ldots$ such that for every $n \in \mathbb{N}$, $a_n
\in \{ 0, 1 \}$ satisfies $a_n = 1$ iff $n \in P$, i.e., if $P(n)$
holds.

Encodings for the values of first-order variables $x_1$, \ldots, $x_k$
and second-order variables $X_1$, \ldots, $X_{\ell}$ can be combined
into a single word over the alphabet $\Sigma = \{ 0, 1 \}^{k+\ell}$: A
word $w \in \Sigma^{\omega}$ encodes an interpretation $I$ for those
variables iff $w = (a_{0,1}, \ldots, a_{0,k+\ell})(a_{1,1}, \ldots,
a_{1,k+\ell}) \ldots$, where for each $i \in [1, k]$,
$a_{0,i}a_{1,i}\ldots$ encodes $I(x_i)$, and for each $j \in [1,
  \ell]$, $a_{0,k+j}a_{1,k+j}\ldots$ encodes $I(X_j)$.  Note that not all
infinite words over $\Sigma$ form valid encodings: For each
first-order variable $x_i$, an encoding must contain exactly one occurrence
of the symbol $1$ for the $i$-th component of its tuple symbols.
Assuming that the set of variables is clear from the context, we write
$e(I)$ for the encoding of $I$ with respect
to those variables.

\subsection{Automata Recognizing Sets of Models}

Let $S$ be a set of interpretations for $k$ first-order and $\ell$
second-order variables. The set of encodings of the elements of $S$
forms a language $L$ over the alphabet $\{0, 1\}^{k+\ell}$. If this
language is $\omega$-regular, then we say that an automaton $\cal A$
that accepts $L$ \textit{recognizes}, or \textit{represents}, the set
$S$. Such an automaton can be viewed as having $k + \ell$ input tapes
reading symbols from $\{ 0, 1 \}$, each of these tapes being
associated to a variable. Equivalently, we can write the label of a
transition $(q_1, (a_1, \ldots, a_{k+\ell}), q_2) \in \Delta$ as
$\stackrel{(a_{k+1}, \ldots, a_{k+\ell})}{V}$ where $V$ is the set of
the variables $x_i$, with $i \in [1, k]$, for which $a_i = 1$.  In
other words, each transition label distinct from $\varepsilon$
specifies the set of first-order variables whose value corresponds to
this transition, and provides one symbol for each second-order
variable. For each $x_i \in V$, we then say that $x_i$ is
\textit{associated} to the transition.  Note that every transition for
which $V \neq \emptyset$ can only be followed at most once in an
accepting run. Any automaton recognizing a set of valid encodings can
therefore easily be transformed into one in which such transitions do
not appear in cycles, and that accepts the same language.

\begin{figure}
\centering
\vspace*{-3ex}
\begin{subfigure}[t]{0.45\textwidth}
\begin{tikzpicture}
  \node (q0) [state, initial, initial text = {}] {$q_0$};
  \node (q1) [state, right = of q0] {$q_1$};
  \node (q2) [state, accepting, right = of q1] {$q_2$};

  \path [-stealth, thick]
     (q0) edge [loop above] node
         {$\begin{array}{c}~\\(0), (1)\end{array}$} ()
     (q1) edge [loop above] node
         {$\begin{array}{c}~\\(0),(1)\end{array}$} ()
     (q2) edge [loop above] node
         {$\begin{array}{c}~\\(0),(1)\end{array}$} ()
     (q0) edge [below, inner sep = 10pt] node
         {$\begin{array}{c}(0), (1)\\\{ x_1 \}\end{array}$} (q1)
     (q1) edge [below, inner sep = 10pt] node
         {$\begin{array}{c}(1)\\\{ x_2 \}\end{array}$} (q2);
\end{tikzpicture}
\vspace*{-1ex}
\caption{$x_1 < x_2 \wedge X_1(x_2)$}
\label{fig:automata-formula}
\end{subfigure}
\begin{subfigure}[t]{0.45\textwidth}
\begin{tikzpicture}
  \node (q0) [state, initial, initial text = {}] {$q_0$};
  \node (q1) [state, right = of q0] {$q_1$};
  \node (q2) [state, accepting, right = of q1] {$q_2$};

  \path [-stealth, thick]
     (q0) edge [loop above] node
         {$\begin{array}{c}~\\(0), (1)\end{array}$} ()
     (q1) edge [loop above] node
         {$\begin{array}{c}~\\(0),(1)\end{array}$} ()
     (q2) edge [loop above] node
         {$\begin{array}{c}~\\(0),(1)\end{array}$} ()
     (q0) edge [below, inner sep = 10pt] node
         {$\begin{array}{c}(0), (1)\\\{ x_1 \}\end{array}$} (q1)
     (q1) edge [below, inner sep = 10pt] node
         {$\begin{array}{c}(1)\\~\end{array}$} (q2);
\end{tikzpicture}
\vspace*{-1ex}
\caption{$\exists x_2 (x_1 < x_2 \wedge X_1(x_2))$}
\label{fig:automata-formula-exists}
\end{subfigure}
\caption{Automata recognizing sets of models.}
\label{fig:automata-notations}
\end{figure}
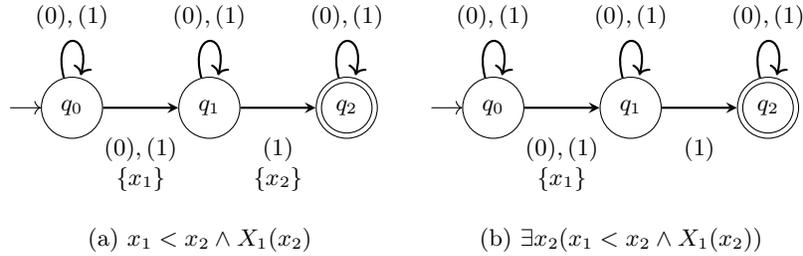

An example of an automaton recognizing the set of models of the
formula $\varphi(x_1, x_2, X_1) = x_1 < x_2 \,\wedge\, X_1(x_2)$ is
given in Figure~\ref{fig:automata-formula}.
For the sake of clarity, labels of transitions sharing the same
origin and destination are grouped together, and empty sets of variables are
omitted.

\subsection{Decision Procedure}
\label{sec:sub-decision}

For the automata-based approach to be applicable, it must be possible
to construct elementary automata recognizing the models of atomic
formulas. This is clearly the case for atoms of the form $x_i = x_j$,
$X_i = X_j$ and $X_i(x_j)$, and this property must also hold for each
relation $R_i$ that belongs to the structure of interest; in other
words, the atomic formula $R_i(x_1, x_2, \ldots, x_{\alpha_i})$ must
admit a set of models whose encoding is $\omega$-regular.  With the
positional encoding of natural numbers, this is the case in particular
for the order relation $x_i < x_j$ and the successor relation $x_j =
x_i + 1$. Note that one can easily add supplementary variables to an
automaton, by inserting a new component in the tuples of its alphabet,
and making this component read a symbol 1 at any single position of a
run for first-order variables, and any symbol at any position for
second-order ones. Reordering the variables is a similarly immediate
operation.

After automata recognizing the models of atomic formulas have been
obtained, the next step consists in combining them recursively by
following the syntactic structure of the formula to be decided.
Let us denote by $L_\varphi$ the language
of encodings of all the models of a formula $\varphi$, i.e.,
$L_\varphi = \{ e(I) \mid I \models \varphi \}$.

For the Boolean operator $\wedge$, we have $L_{\varphi_1 \wedge
  \varphi_2} = L_{\varphi_1} \cap L_{\varphi_2}$, where $\varphi_1$
and $\varphi_2$ are formulas over the same free variables. Similarly,
we have $L_{\varphi_1 \vee \varphi_2} = L_{\varphi_1} \cup
L_{\varphi_2}$. The case of the complement operator $\neg$ is slightly
more complicated, since the complement of a language of encodings
systematically contains words that do not validly encode an
interpretation.  The set of models of a formula $\neg
\varphi$ is encoded by the language $\overline{L_\varphi}\cap
L_{\mathit{valid}}$, where $L_{\mathit{valid}}$ is the language of all
valid encodings consistent with the free variables of $\varphi$.
It is easily seen that this language is $\omega$-regular.

It remains to compute the effect of quantifiers. The language
$L_{\exists x_i\varphi}$ can be derived from $L_{\varphi}$ by
removing the $i$-th component from each tuple symbol, i.e., by
applying a mapping $\Pi_{\neq i}: \Sigma^{k+\ell} \rightarrow
\Sigma^{k+\ell-1}: (a_1, \ldots, a_{k+\ell}) \mapsto (a_1, \ldots,
a_{i-1}, a_{i+1}, \ldots, a_{k+\ell})$ to each symbol of the alphabet.
Indeed, the models of $\exists x_i\, \varphi$ correspond exactly to the
models of $\varphi$ in which the variable $x_i$ is removed. In the rest
of this paper, we will use the notation $\Pi_{\neq i}(w)$, where $w$ is a finite
or infinite word, to express the result of applying $\Pi_{\neq i}$ to each
symbol in $w$. If $L$ is a language, then we write $\Pi_{\neq i}(L)$ for
the language $\{ \Pi_{\neq i}(w) \mid w \in L \}$.

Finally, universal quantification can be reduced to existential
quantification:
For computing $L_{\forall x_i\varphi}$, we use the equivalence $\forall
x_i\, \varphi \,\equiv \, \neg \exists x_i\, \neg \varphi$ which yields
$L_{\forall x_i\varphi} = \overline{L_{\exists x_i \neg\varphi}}\cap
L_{\mathit{valid}}$.

\subsection{Operations over Automata}
\label{sec:operations}

We now discuss how the operations over languages mentioned in
Section~\ref{sec:sub-decision} can be computed over infinite-word
automata.  Given automata ${\cal A}_1$ and ${\cal A}_2$, automata
${\cal A}_1 \cap {\cal A}_2$ and ${\cal A}_1 \cup {\cal A}_2$
accepting respectively $L({\cal A}_1) \cap L({\cal A}_2)$ and $L({\cal
  A}_1) \cup L({\cal A}_2)$ can be obtained by the so-called product
construction. The idea consists in building an automaton $\cal A$ that
simulates the combined behavior of ${\cal A}_1$ and ${\cal A}_2$ on
identical input words. The states of $\cal A$ need to store additional
information about the accepting states that are visited in ${\cal
  A}_1$ and ${\cal A}_2$. For ${\cal A}_1 \cap {\cal A}_2$, one
ensures that each accepting run of $\cal A$ correspond to an accepting
run in both ${\cal A}_1$ and ${\cal A}_2$. For ${\cal A}_1 \cup {\cal
  A}_2$, the condition is that the run should be accepting in ${\cal
  A}_1$ or ${\cal A}_2$, or both. A complete description of the
product construction for Büchi automata is given in~\cite{Thomas90}.

Modifying the alphabet of an automaton in order to implement the
effect of an existential quantification is a simple operation.  As an
example, Figure~\ref{fig:automata-formula-exists} shows an automaton
recognizing the set of models of $\exists x_2 (x_1 < x_2 \wedge
X_1(x_2))$, obtained by removing all occurrences of the variable $x_2$
from transition labels.  Testing whether the language accepted by an
automaton is not empty amounts to checking the existence of a
reachable cycle that visits at least one accepting state, which is
simple as well.

The only problematic operation is complementation, which consists in
computing from an automaton $\cal A$ an automaton that accepts the
language $\overline{L({\cal A})}$. Although it preserves
$\omega$-regularity, this operation is difficult to perform on Büchi
automata~\cite{Safra88,Vardi07}. In the context of our decision
procedure, it is only useful for applying universal quantifiers.
Indeed, other instances of the negation operator in formulas can be
pushed inwards until they are applied to atomic formulas, and it is
easy to construct the complement of the elementary automata
recognizing the models of those atomic formulas, provided that for
each relation $R_i$ in the structure of interest, an automaton
recognizing $\{ (x_1, \ldots, x_{\alpha_i}) \in \mathbb{N}^{\alpha_i}
\mid \neg R_i(x_1, \ldots, x_{\alpha_i})\}$ is available.  In order to
eliminate the need for complementation, we develop in the next section
a direct algorithm for computing the effect of universal quantifiers
on automata recognizing sets of models.

\section{Universal Quantification}
\label{sec:universal-quant}

\subsection{Principles}

Let $\varphi(x_1, \ldots, x_k, X_1, \ldots, X_{\ell})$, with $k > 0$
and $\ell \geq 0$, be a formula. Our goal is to compute an automaton
${\cal A}'$ accepting $L_{\forall x_i \varphi}$, given an automaton
$\cal A$ accepting $L_{\varphi}$ and $i \in [1, k]$.

By definition of universal quantification, we have $I \models \forall
x_i\, \varphi$ iff $I[x_i = n] \models \varphi$ holds for every $n \in
\mathbb{N}$. In other words, $L_{\forall x_i \varphi}$ contains $e(I)$
iff $L_{\varphi}$ contains $e(I[x_i=n])$ for every $n \in
\mathbb{N}$. Conceptually, we can then obtain $L_{\forall x_i
  \varphi}$ by defining for each $n \in \mathbb{N}$ the language $S_n
= \{ e(I) \mid e(I[x_i=n]) \in L_{\varphi} \}$, which yields
$L_{\forall x_i \varphi} = \bigcap_{n \in \mathbb{N}} S_n$.

An automaton ${\cal A}'$ accepting $L_{\forall x_i \varphi}$ can be
obtained as follows. Each language $S_n$, with $n \in \mathbb{N}$, is
accepted by an automaton ${\cal A}_n$ derived from $\cal A$ by
restricting the transitions associated to $x_i$ to be followed only
after having read exactly $n$ symbols. In other words, the accepting
runs of ${\cal A}_n$ correspond to the accepting runs of $\cal A$ that
satisfy this condition.  After imposing this restriction, the variable
$x_i$ is removed from the set of variables managed by the automaton,
i.e., the operator $\Pi_{\neq i}$ is applied to the language that this
automaton accepts, so as to get $S_n = L({\cal A}_n)$. The automaton
${\cal A}'$ then corresponds to the infinite intersection product of
the automata ${\cal A}_n$ for all $n \in \mathbb{N}$, i.e., an
automaton that accepts the infinite intersection $\bigcap_{n \in
  \mathbb{N}} S_n$.  We show in the next section how to build ${\cal
  A}'$ by means of a finite computation.

\subsection{Construction}
\label{sec:construction}

The idea of the construction is to make ${\cal A}'$ simulate the join
behavior of the automata ${\cal A}_n$, for all $n \in \mathbb{N}$, on
the same input words. This can be done by making each state of ${\cal
  A}'$ correspond to one state $q_n$ in each ${\cal A}_n$, i.e., to an
infinite tuple $(q_0, q_1, \ldots)$.  By definition of ${\cal A}_n$,
there exists a mapping $\mu: Q_n \rightarrow Q$, where $Q_n$ and $Q$
are respectively the sets of states of ${\cal A}_n$ and $\cal A$, such
that whenever a run of ${\cal A}_n$ visits $q_n$, the corresponding
run of ${\cal A}$ on the same input word visits $\mu(q_n)$.

If two automata ${\cal A}_{n_1}$ and ${\cal A}_{n_2}$, with $n_1, n_2
\in \mathbb{N}$, are (respectively) in states $q_{n_1}$ and $q_{n_2}$
such that $\mu(q_{n_1}) = \mu(q_{n_2})$, then they share the same
future behaviors, except for the requirement to follow a transition
associated to $x_i$ after having read (respectively) $n_1$ and $n_2$
symbols. It follows that the states of ${\cal A}'$ can be
characterized by sets of states of $\cal A$: The infinite tuple $(q_0,
q_1, \ldots)$ is described by the set $\{ \mu(q_i) \mid i \in
\mathbb{N} \}$. Each element of this set represents the current state
of one or several automata among the ${\cal A}_{n}$. This means that
the number of these automata that are in this current state is not
counted. We will establish that this abstraction is precise and leads
to a correct construction.

During a run of ${\cal A}'$, each transition with a label other than
$\varepsilon$ must correspond to a transition reading the same symbol
in every automaton ${\cal A}_n$, which in turn can be mapped to a
transition of $\cal A$. In the automaton ${\cal A}_n$ for which $n$ is
equal to the number of symbols already read during the run, this
transition of $\cal A$ is necessarily associated to $x_i$, by definition
of ${\cal A}_n$. It follows that every transition of ${\cal A}'$ with
a non-empty label is characterized by a set of transitions of $\cal
A$, among which one of them is associated to $x_i$.

We are now ready to describe formally the construction of ${\cal A}'$,
leaving for the next section the problem of determining which of its
runs should be accepting or not: From the automaton
${\cal A} = (\Sigma, Q, \Delta, Q_0, F)$, we construct
${\cal A}' = (\Sigma',
Q', \Delta', Q'_0, F')$ such that
\begin{itemize}
\item
$\Sigma' = \Pi_{\neq i}(\Sigma)$.
\item
$Q' = 2^Q \setminus \{ \emptyset \}$.
\item
$\Delta'$ contains
\begin{itemize}
\item
the transitions $(q'_1, (a'_1, \ldots, a'_{k+\ell-1}), q'_2)$
for which there exists a set $T \subseteq \Delta$ that satisfies the
following conditions:
\begin{itemize}
\item
$q'_1 = \{ q_1 \mid (q_1, (a_1, \ldots, a_{k+\ell}), q_2) \in T\}$.
\item
$q'_2= \{ q_2 \mid (q_1, (a_1, \ldots, a_{k+\ell}), q_2) \in T\}$.
\item
For all $(q_1, (a_1, \ldots, a_{k+\ell}), q_2) \in T$, one has
$a'_j = a_j$ for all $j \in [1, i-1]$, and
$a'_j = a_{j+1}$ for all $j \in [i, k+\ell-1]$.
\item
There exists exactly one $(q_1, (a_1, \ldots, a_{k+\ell}), q_2) \in T$
such that $a_i = 1$.
\end{itemize}
\item
the transitions $(q'_1, \varepsilon, q'_2)$ for which there exists
a transition $(q_1, \varepsilon, q_2) \in \Delta$
such that
\begin{itemize}
\item $q_1 \in q'_1$.
\item $q'_2 = q'_1 \cup \{ q_2 \}$ or
$q'_2 = (q'_1 \setminus \{ q_1 \}) \cup \{ q_2 \}$.
\end{itemize}
\end{itemize}
\item
$Q'_0 = 2^{Q_0} \setminus \{ \emptyset \}$.
\item
$F' = Q'$ for now. The problem of characterizing more finely the accepting
runs will be addressed in the next section.
\end{itemize}

The rule for the transitions $(q'_1, (a'_1, \ldots, a'_{k+\ell-1}),
q'_2)$ ensures that for each $q_1 \in q'_1$, each automaton ${\cal A}_n$
that is simulated by ${\cal A}'$ has the choice of following any
possible transition originating from $q_1$ that has a label consistent
with $(a'_1, \ldots, a'_{k+\ell-1})$. One such automaton must
nevertheless follow a transition associated to the quantified variable $x_i$.
The rule for the transitions $(q'_1, \varepsilon, q'_2)$ expresses that
one automaton ${\cal A}_n$, or any number of identical copies of this
automaton, must follow a transition labeled by
$\varepsilon$, while the other automata stay in their current state.

As an example, applying this construction to the automaton in
Figure~\ref{fig:automata-formula-exists}, as a first step of the
computation of a representation of $\forall x_1 \exists x_2 (x_1 <
x_2 \wedge X_1(x_2))$, yields the automaton given in
Figure~\ref{fig:example-universal-1}. For the sake of clarity,
unreachable states and states from which the accepted language
is empty are not depicted.

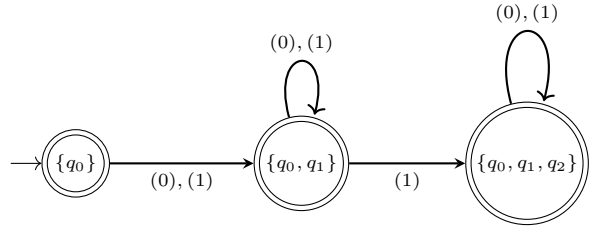
\begin{figure}
\centering
\vspace*{-3ex}
\begin{tikzpicture}
  \node (q0) [state, accepting, initial, initial text = {}] at (0,3)
      {\scriptsize $\{q_0\}$};
  \node (q1) [state, accepting] at (3,3) {\scriptsize $\{q_0, q_1\}$};
  \node (q2) [state, accepting] at (6,3) {\scriptsize $\{q_0, q_1, q_2\}$};

  \path [-stealth, thick]
     (q1) edge [loop above] node
         {\scriptsize $(0), (1)$} ()
     (q2) edge [loop above] node
         {\scriptsize $(0), (1)$} ()
     (q0) edge [below] node {\scriptsize $(0), (1)$} (q1)
     (q1) edge [below] node {\scriptsize $(1)$} (q2);
\end{tikzpicture}
\vspace*{-1ex}
\caption{First step of construction for $\forall x_1 \exists x_2
(x_1 < x_2 \wedge X_1(x_2))$.}
\label{fig:example-universal-1}
\vspace*{-5ex}
\end{figure}

\subsection{A Criterion for Accepting Runs}
\label{sec:criterion}

The automaton ${\cal A}'$ defined in the previous section simulates an
infinite combination of automata ${\cal A}_n$, for all $n \in
\mathbb{N}$. By construction, every accepting run of
this infinite combination corresponds to a run of ${\cal A}'$.

The reciprocal property is not true, in the sense that there may exist
a run of ${\cal A}'$ that does not match an accepting run of the
infinite combination of automata ${\cal A}_n$. Consider for instance a
run of the automaton in Figure~\ref{fig:example-universal-1} that ends
up cycling in the state $\{ q_0, q_1, q_2 \}$, reading $0^{\omega}$
from that state. Recall that for this example, the automaton $\cal A$
that undergoes the universal quantification operation is the one given
in Figure~\ref{fig:automata-formula-exists}.  The run that we have
considered can be followed in ${\cal A}'$, but cannot be accepting in
every ${\cal A}_n$. Indeed, in this example, the transition of ${\cal
  A}_n$ reading the $(n+1)$-th symbol of the run corresponds, by
definition of this automaton, to the transition of ${\cal A}$ that is
associated to the quantified variable $x_1$.  By the structure of
$\cal A$, this transition is necessarily followed later in any
accepting run by one that reads the symbol $1$, which implies that no
word of the form $u \cdot 0^{\omega}$, with $u \in \{ 0,1 \}^*$, can
be accepted by a run of ${\cal A}_n$ such that $n \geq |u|$. This
represents the fact that the words accepted by all ${\cal A}_n$
correspond to the encodings of predicates that are true infinitely
often.

One thus needs a criterion for characterizing the runs of ${\cal A}'$
that correspond to combinations of accepting runs in all automata
${\cal A}_n$.

It is known~\cite{McNaughton66} that two $\omega$-regular languages
over the alphabet $\Sigma$ are equal iff they share the same set of
\textit{ultimately periodic} words, i.e., words of the form $u \cdot
v^{\omega}$ with $u \in \Sigma^*$ and $v \in \Sigma^+$. It follows
that it is sufficient to characterize the accepting runs of ${\cal
  A}'$ that read ultimately periodic words. The automaton ${\cal A}'$
accepts a word $u \cdot v^{\omega}$ iff every ${\cal A}_n$, with $n
\in \mathbb{N}$, admits an accepting run that reads this word. Note
that such a run also matches a run of $\cal A$, and that this run of
$\cal A$ always ends up following a cycle from an accepting state to
itself.

Our solution takes the following form. For each state $q$ of $\cal A$,
we define a language $U_q \subseteq \Sigma^+$ of non-empty words $u$ such that
$\cal A$ accepts $u^{\omega}$ from $q$, after dismissing the input tape
associated to the quantified variable $x_i$. The alphabet $\Sigma$ is
thus equal to $\{0, 1\}^{k + \ell - 1}$.
Remember that each state $q'$ of ${\cal A}'$ is defined as a subset of
states of $\cal A$, corresponding to the current states in the
combination of copies of $\cal A$ that are jointly simulated by ${\cal
  A}'$.  In order for the word $u^{\omega}$ to be accepted by ${\cal
  A}'$ from $q'$, it should therefore be accepted by $\cal A$ from
each state $q \in q'$, i.e., $u $ must belong to all the languages
$U_q$ such that $q \in q'$.

It must also be possible to read $u^{\omega}$ from the state $q'$ of
${\cal A}'$. We impose a stronger condition, by requiring that there
exists a cycle from $q'$ to itself labeled by $u$. This condition
leads to a correct acceptance criterion.

In summary, the language $U'_{q'} =  L({\cal A}', q', q') \cap \bigcap_{q \in q'}
U_q$ characterizes the words $u$ such that
$u^{\omega}$ must be accepted from the state $q'$ of ${\cal A}'$.
 Note that for
this property to hold, it is not necessary for the language $U_q$ to
contain all words $u$ such that $u^{\omega} \in \Pi_{\neq i}(L({\cal A}, q))$, but
only some number of copies $u^p$, where $p > 0$ is bounded, of each
such $u$. In other words, the finite words $u$ whose infinite
repetition is accepted from $q$ do not have to be the shortest
possible ones.

Once the language $U'_{q'}$ has been obtained, we build a
\textit{widget}, in the form of an infinite-word automaton accepting
$(U'_{q'})^{\omega}$, along the state $q'$ of ${\cal A}'$, and add a
transition labeled by $\varepsilon$ from $q'$ to the initial state of
this widget. This ensures that every path that ends up in $q'$ can be
suitably extended into an accepting run.  Such a widget does not have
to be constructed for every state $q'$ of ${\cal A}'$: Since the goal
is to accept from $q'$ words of the form $u^{\omega}$, we can require
that at least one state $q \in q'$ is accepting in $\cal A$. We then
only build widgets for the states $q'$ that satisfy this requirement.

\subsection{Computation Steps}
\label{sec:implementation}

The procedure for modifying ${\cal A}'$ in order to make it accept the
runs that match those of the infinite combination of automata ${\cal
  A}_n$, outlined in the previous section, can be carried out by
representing the regular languages $U_q$ and $U'_{q'}$ by finite-state
automata.  The construction proceeds as follows:

\begin{enumerate}
\item
For each state $q \in Q$ of $\cal A$, build a finite-word automaton
${\cal A}_{q}$ that accepts all the non-empty words $u$ for which
there exists a path $q \stackrel{v}{\rightarrow} q$ of $\cal A$ that
visits at least one accepting state $q_F \in F$, such that $u =
\Pi_{\neq i}(v)$.
This automaton can be constructed in a similar way as one accepting
$\Pi_{\neq i}(L({\cal A}, q, q))$ (cf. Sections~\ref{sec:automata}
and~\ref{sec:operations}), keeping one additional bit of information
in its states for determining whether an accepting state has already
been visited or not.

\item For each pair of states $q_1, q_2 \in Q$ of $\cal A$, build a finite-word
  automaton ${\cal A}_{q_1, q_2}$ accepting the language $\Pi_{\neq i}(L({\cal A},
  q_1, q_2))$ (cf. Sections~\ref{sec:automata}
and~\ref{sec:operations}).
\item
For each state $q\in Q$ of $\cal A$, build an automaton
${\cal A}_{U_{q}} = \bigcup_{r \in Q} \left( {\cal A}_{q,r} \cap {\cal
  A}_{r}\right)$ accepting the finite-word language $U_{q}$.

\item
For each state $q'$ of ${\cal A}'$ such that $q' \cap F \neq
\emptyset$, where $F$ is the set of accepting states of $\cal A$,
build a finite-word automaton ${\cal A}'_{U'_{q'}} = {\cal A}'_{q'}
\cap\, \bigcap_{q \in q'} {\cal A}_{U_q}$ accepting $U'_{q'}$, where
${\cal A}'_{q'}$ is an automaton accepting $L({\cal A}', q', q')$
(cf. Section~\ref{sec:automata}).

\item
Then, turn each automaton ${\cal A}'_{U'_{q'}}$ into an
infinite-word automaton ${\cal A}'_{(U'_{q'})^{\omega}}$ accepting
$(U'_{q'})^{\omega}$:
\begin{enumerate}
\item
Create a new state $q'_{\mathit{repeat}}$.
\item
Add a transition $(q'_{\mathit{repeat}}, \varepsilon, q_0)$ for each
initial state $q_0$, and a transition $(q_F, \varepsilon, q'_{\mathit{repeat}})$
for each accepting state $q_F$, of ${\cal A}'_{U'_{q'}}$.
\item
Make $q'_{\mathit{repeat}}$ the only initial and accepting state of
${\cal A}'_{(U'_{q'})^{\omega}}$.
\end{enumerate}
\item
For each state $q'$ of ${\cal A}'$ considered at Step~4, add the
widget ${\cal A}'_{U'_{q'}}$ alongside $q'$, by incorporating its sets
of states and transitions into those of ${\cal A}'$, and adding a
transition $(q', \varepsilon, q'_{\mathit repeat})$.  In the resulting
automaton, mark as the only accepting states the states
$q'_{\mathit{repeat}}$ of all widgets.
\end{enumerate}
We call ${\cal A}''$ the automaton constructed by this procedure.
This automaton accepts the language $L_{\forall x_i \varphi}$.

\subsection{Illustration}

We illustrate the construction described in Section~\ref{sec:implementation}
on the automaton ${\cal A}'$ in Figure~\ref{fig:example-universal-1},
obtained after the first step of the universal quantification
procedure applied to the automaton ${\cal A}$ in
Figure~\ref{fig:automata-formula-exists}. We assume that the reader is
familiar with the notation of regular and $\omega$-regular languages
by regular expressions.

We obtain automata accepting the following languages at each step of
the procedure:

\begin{enumerate}
\item
$L({\cal A}_{q_0}) = L({\cal A}_{q_1}) = \emptyset$ and
$L({\cal A}_{q_2}) = (0 + 1)^+$.
\item
$L({\cal A}_{q_0, q_0}) = L({\cal A}_{q_1, q_1}) =
 L({\cal A}_{q_2, q_2}) = (0 + 1)^*$,
$L({\cal A}_{q_0, q_1}) = (0 + 1)^+$,
$L({\cal A}_{q_0, q_2}) = (0 + 1)^+ \,1\, (0 + 1)^*$,
$L({\cal A}_{q_1, q_2}) = (0 + 1)^* \,1\, (0 + 1)^*$,
and
$L({\cal A}_{q_j q_{j'}}) = \emptyset$ for all other pairs $(q_j,q_{j'})$ of states.
\item
$L({\cal A}_{U_{q_0}}) = (0 + 1)^+ \,1\, (0 + 1)^*$,
$L({\cal A}_{U_{q_1}}) = (0 + 1)^* \,1\, (0 + 1)^*$, and
$L({\cal A}_{U_{q_2}}) = (0 + 1)^+$.
\item
$L\left({\cal A}'_{U'_{\{q_1,
    q_2\}}}\right) = \emptyset$ and
$L\left({\cal A}'_{U'_{\{q_0, q_1, q_2\}}}\right) =
(0 + 1)^+ \,1\, (0 + 1)^*$.
\item
The widget ${\cal A}'_{U'_{\{q_0, q_1, q_2\}}}$ is given in Figure~\ref{fig:widget}.
\item
The resulting automaton ${\cal A}''$ is shown in
Figure~\ref{fig:example-universal-2}. For the sake of clarity,
the states from which the accepted language is empty
have been removed.
\end{enumerate}

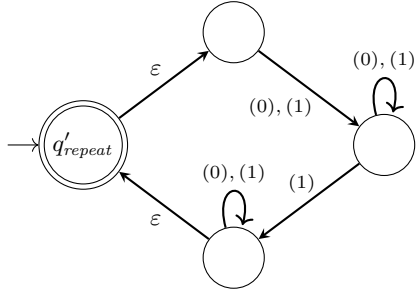
\begin{figure}
\centering
\begin{tikzpicture}

  \node (q0) [state] at (2,4.5) {};
  \node (q1) [state] at (4,3) {};
  \node (q2) [state] at (2,1.5) {};
  \node (q3) [state, initial, initial text = {}, accepting] at (0,3)
      {$q'_{\mathit{repeat}}$};

  \path [-stealth, thick]
      (q1) edge [loop above] node {\scriptsize $(0), (1)$} ()
      (q2) edge [loop above] node {\scriptsize $(0), (1)$} ()
      (q0) edge [below] node {\scriptsize $(0), (1)$~~~~~~~~~} (q1)
      (q1) edge [above] node {\scriptsize $(1)$~~~} (q2)
      (q2) edge [below] node { $\varepsilon$~~ } (q3)
      (q3) edge [above] node { $\varepsilon$~~ } (q0);
\end{tikzpicture}
\caption{Widget for the state $\{ q_0, q_1, q_2 \}$.}
\label{fig:widget}
\end{figure}

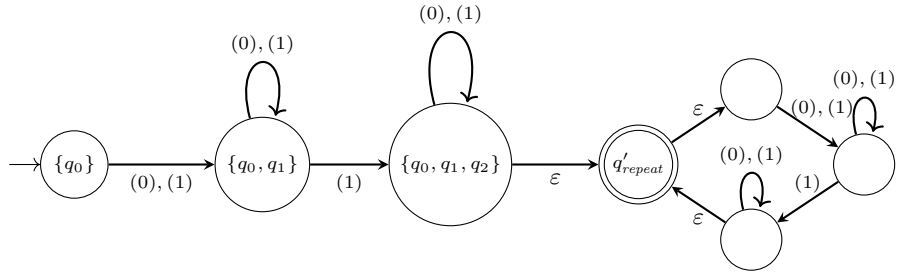
\begin{figure}
\centering
\begin{tikzpicture}
  \node (q0) [state, initial, initial text = {}] at (0,3)
      {\scriptsize $\{q_0\}$};
  \node (q1) [state] at (2.5,3) {\scriptsize$\{q_0, q_1\}$};
  \node (q2) [state] at (5,3) {\scriptsize$\{q_0, q_1, q_2\}$};

  \node (q5) [state] at (9,4) {};
  \node (q6) [state] at (10.5,3) {};
  \node (q7) [state] at (9,2) {};
  \node (q8) [state, accepting] at (7.5,3)
      {\scriptsize $q'_{\mathit{repeat}}$};

  \path [-stealth, thick]
     (q1) edge [loop above] node
         {\scriptsize $(0), (1)$} ()
     (q2) edge [loop above] node
         {\scriptsize $(0), (1)$} ()
     (q0) edge [below] node {\scriptsize $(0), (1)$} (q1)
     (q1) edge [below] node {\scriptsize $(1)$} (q2)

     (q6) edge [loop above] node {\scriptsize $(0), (1)$} ()
     (q7) edge [loop above] node {\scriptsize $(0), (1)$} ()
     (q5) edge [above] node {\scriptsize ~~~~$(0), (1)$} (q6)
     (q6) edge [above] node {\scriptsize $(1)$} (q7)
     (q7) edge [below] node {$\varepsilon$} (q8)
     (q8) edge [above] node {$\varepsilon$ } (q5)
     (q2) edge [below] node{$\varepsilon$ } (q8);
\end{tikzpicture}
\vspace*{-1ex}
\caption{Automaton recognizing the set of models of $\forall x_1 \exists x_2
(x_1 < x_2 \wedge X_1(x_2))$.}
\label{fig:example-universal-2}
\end{figure}

\subsection{Proof of Correctness}

Let us show that the automaton ${\cal A}''$ constructed according to
the procedure developed in Section~\ref{sec:implementation} simulates
exactly the infinite intersection of the automata ${\cal A}_n$, for
all $n \geq 0$, as defined in Section~\ref{sec:construction}.  We
first establish that every word accepted by ${\cal A}''$ is also
accepted by ${\cal A}_n$ for each $n \geq 0$, after dismissing the input
tape associated to the quantified variable. Note that, since we know
that applying a first-order universal quantifier to an automaton
recognizing a set of models preserves $\omega$-regularity
(cf. Section~\ref{sec:operations}), and two $\omega$-regular languages
are equal iff they share the same ultimately periodic
words~\cite{McNaughton66}, it is sufficient to consider ultimately
periodic words.

\begin{theorem}
For every ultimately periodic word $w'' \in L({\cal A}'')$ and $n \geq 0$,
there exists an accepting run $\rho$ of $\cal A$ such that
\begin{itemize}
\item
the word $w$ read by $\rho$ satisfies $\Pi_{\neq i}(w) = w''$, and
\item
the $(n+1)$-th transition with a non-empty label followed by $\rho$
is associated to the quantified variable $x_i$. (In other words, such
a run $\rho$ also exists for the automaton ${\cal A}_n$.)
\end{itemize}
\end{theorem}
\begin{proof}
The only accepting states of ${\cal A}''$ are the states $q'_{\mathit{repeat}}$,
and those are not reachable from each other. Thus, for every run of ${\cal A}''$
that accepts $w''$, there exists a single state $q'_{\mathit{repeat}}$ that this
run visits infinitely often. It follows that there exists a run $\rho''$
of ${\cal A}''$ accepting $w''$ that is of the form
\[
\rho'':\, q'_0 \stackrel{u_1}{\rightarrow}
q' \stackrel{\varepsilon}{\rightarrow}
q'_{\mathit{repeat}} \stackrel{u_2}{\rightarrow}
q'_{\mathit{repeat}} \stackrel{u_3}{\rightarrow}
q'_{\mathit{repeat}} \stackrel{u_3}{\rightarrow}
q'_{\mathit{repeat}} \stackrel{u_3}{\rightarrow}
\cdots,
\]
where $q'_0$ is an initial state of ${\cal A}''$, and $w'' = u_1 \, u_2
\, (u_3)^{\omega}$ with $|u_2| > 0$ and $|u_3| > 0$.

By construction of the widget associated to $q'$, if there exists a
non-empty word $u$ such that ${\cal A}''$ admits the path
$q'_{\mathit{repeat}} \stackrel{u}{\rightarrow} q'_{\mathit{repeat}}$,
then the path $q' \stackrel{u}{\rightarrow} q'$ exists as well in
${\cal A}'$, and therefore in ${\cal A}''$. (Securing this property
was the motivation behind the restriction discussed in the definition
of $U'_{q'}$ in Section~\ref{sec:criterion}.) Therefore, ${\cal A}''$
admits the paths $q' \stackrel{u_2}{\rightarrow} q'$ and $q'
\stackrel{u_3}{\rightarrow} q'$, and there exists a run $\rho''_1$ of
${\cal A}''$ accepting $w$ that is of the form
\[
\rho_1'':\, q'_0 \stackrel{u_1}{\rightarrow}
q' \stackrel{u_2}{\rightarrow}
q'\stackrel{u_3}{\rightarrow}
q' \stackrel{\varepsilon}{\rightarrow}
q'_{\mathit{repeat}} \stackrel{u_3}{\rightarrow}
q'_{\mathit{repeat}} \stackrel{u_3}{\rightarrow}
q'_{\mathit{repeat}} \stackrel{u_3}{\rightarrow}
\cdots.
\]

For every $p > 0$, the path $q' \stackrel{u_3}{\rightarrow} q'$ that appears
in this run can be repeated $p$ times, which results in a path
\[
\rho_p'':\, q'_0 \stackrel{u_1}{\rightarrow}
q' \stackrel{u_2}{\rightarrow}
q'\stackrel{(u_3)^p}{\rightarrow}
q' \stackrel{u_3}{\rightarrow}
q'_{\mathit{repeat}} \stackrel{u_3}{\rightarrow}
q'_{\mathit{repeat}} \stackrel{u_3}{\rightarrow}
q'_{\mathit{repeat}} \stackrel{u_3}{\rightarrow}
\cdots
\]
that accepts $w''$. We then choose $p$ such that $|u_1 \, u_2 \,
(u_3)^p| > n$, in order to ensure that the $(n + 1)$-th
symbol read by $\rho_p''$ corresponds to a transition located in
${\cal A}'$ and not in any widget.

By construction of ${\cal A}'$ (cf. Section~\ref{sec:construction}),
there exists a path $\pi$ of ${\cal A}$ of the form $\pi: q_0
\stackrel{u}{\rightarrow} q_i$, where $q_0$ is an initial state, $q_i
\in q'$, $\Pi_{\neq i}(u) = u_1 \cdot u_2 \cdot (u_3)^p$, and the
$(n+1)$-th symbol read by $\pi$ is associated to the
quantified variable $x_i$.  Furthermore, since the widget associated
to $q'$ admits the path $q'_{\mathit{repeat}}
\stackrel{u_3}{\rightarrow} q'_{\mathit{repeat}}$, we have $u_3 \in
(U_{q_i})^*$, which implies that ${\cal A}$ accepts from its state
$q_i$ a word $v$ such that $\Pi_{\neq i}(v) = (u_3)^{\omega}$. By
appending a run of $\cal A$ accepting $v$ from $q_i$ to the path
$\pi$, we obtain a suitable run $\rho$.
\end{proof}

We now show that every word accepted by the automata ${\cal A}_n$,
for all $n \geq 0$, is also accepted by ${\cal A}''$.

\begin{theorem}
If an ultimately periodic word $w''$ is such that for every $n \geq 0$,
there exists an accepting run $\rho_n$ of $\cal A$ such that
\begin{itemize}
\item
the word $w_n$ read by $\rho_n$ satisfies $\Pi_{\neq i}(w_n) = w''$, and
\item
the $(n+1)$-th transition with a non-empty label followed by $\rho_n$
is associated to the quantified variable $x_i$,
\end{itemize}
then $w'' \in L({\cal A}'')$.
\end{theorem}
\begin{proof}
For a given value of $n$, the run $\rho_n$ takes the form
\[
\rho_n: q_{n,0} \stackrel{a_0}{\rightarrow}
 q_{n,1} \stackrel{a_1}{\rightarrow}
 q_{n,2} \stackrel{a_2}{\rightarrow}
 ~~\cdots~~
 q_{n,n} \!\!\!\stackrel{\begin{array}{c}\scriptstyle a_n\\[-1ex]
 \scriptstyle\{ x_i \}\end{array}}{\rightarrow}
 \!\!\!q_{n,n+1} \stackrel{a_{n+1}}{\rightarrow} ~~\cdots,
\]
where $q_{n,0}$ is initial, and $w'' = a_0 a_1 a_2 \ldots$ with $|a_j|
= 1$ for all $j \geq 0$. The $(n + 1)$-th transition with a non-empty
label followed by this path is associated to the quantified variable
$x_i$. Since $w''$ is ultimately periodic, $w_n$ is ultimately
periodic as well. It follows that $\rho_n$ contains a non-empty cycle
$q \stackrel{u}{\rightarrow} q$, with $q = q_{n, j_n}$ for some $j_n >
0$, that reads a (non necessarily minimal) period $u$ of the infinite
periodic part of $w_n$, and such that $q$ is an accepting state of
$\cal A$. This path can be repeated at will, hence we can assume
w.l.o.g. that $\rho_n$ is of the form
\[
\rho_n: q_{n,0} \stackrel{a_0}{\rightarrow}
 ~~\cdots~~
 q_{n,n} \!\!\!\stackrel{\begin{array}{c}\scriptstyle a_n\\[-1ex]
 \scriptstyle\{ x_i \}\end{array}}{\rightarrow}
 \!\!\!q_{n,n+1} \stackrel{a_{n+1}}{\rightarrow} ~~\cdots~~
q_{n,j_n} \stackrel{u}{\rightarrow} q_{n,j_n+\delta_n}
\stackrel{u}{\rightarrow} q_{n,j_n+2\delta_n}
\stackrel{u}{\rightarrow} ~~\cdots,
\]
where $j_n \geq n + 1$, $\delta_n = |u| > 0$, and
$q_{n,j_n} = q_{n,j_n+\delta_n} = q_{n,j_n+2\delta_n} = \cdots$ is
an accepting state.

By construction of the automaton ${\cal A}'$
(cf. Section~\ref{sec:construction}), the sets of states $\{ q_{0,0},
q_{1,0}, q_{2,0}, \ldots \}$, $\{ q_{0,1}, q_{1,1}, q_{2,1}, \ldots
\}$, $\{ q_{0,2}, q_{1,2}, q_{2,2}, \ldots \}$ correspond
(respectively) to the successive states $q'_0, q'_1, q'_2, \ldots$ of
a run of ${\cal A}'$ reading $w''$. Since $w''$ is ultimately
periodic, this run necessarily contains a path $q'
\stackrel{u_2}{\rightarrow} q'$ that reads a period $u_2$ of the
infinite periodic part of $w''$, with $|u_2| > 0$, such that $q'
\in \{ q'_0, q'_1, \ldots \}$.  This path can be repeated at will,
hence ${\cal A}'$ admits a run $\rho'$ of the form
\[
\rho': q'_0  \stackrel{u_1}{\rightarrow}
q'  \stackrel{u_2}{\rightarrow}
q'  \stackrel{u_2}{\rightarrow}
q'  \stackrel{u_2}{\rightarrow} \cdots,
\]
such that $w'' = u_1 \, (u_2)^{\omega}$. Let us pick an arbitrary
value of $n$, say $n = 0$ for the sake of simplicity. The length of
$u_1$ can be freely increased in $\rho'$, thus we can assume
w.l.o.g. that $u_1$ is such that $|u_1| = j_0 + k \delta_0$ for some
$k \in \mathbb{N}$, which implies $q_{0,j_0 + k \delta_0} \in q'$,
where $q_{0,j_0 + k \delta_0}$ is an accepting state of ${\cal A}$.

Recall that the run $\rho'$ simulates a set of runs of ${\cal A}$ accepting
$w_n$ such that $\Pi_{\neq i} (w_n) = w''$. Whenever $\rho'$ visits the state
$q'$, each of these runs visits some state $q_i \in q'$.
It follows that for each $q_i \in q'$, there exist finite words
$v_{i}, t_{i}$ and a state $q_j \in q'$ such that ${\cal A}$ admits a
run of the form
\[
q_i \stackrel{v_{i}}{\rightarrow}
q_j \stackrel{t_{i}}{\rightarrow}
q_j \stackrel{t_{i}}{\rightarrow}
q_j \stackrel{t_{i}}{\rightarrow}
\cdots,
\]
with $\Pi_{\neq i}(v_{i}) = (u_2)^{m_i}$ and
$\Pi_{\neq i}(t_{i}) = (u_2)^{p_i}$ for some $m_i, p_i > 0$.

By defining $p =  \mathrm{lcm}( \max_{q_i \in q'} \{ m_i \},
\mathrm{lcm}_{q_i \in q'} \{ p_i \})$,
we now have
that
\begin{itemize}
\item
for every $q_i \in q'$, there exists a state $q_j$ of $\cal A$
such that this automaton admits paths $q_i \stackrel{v}{\rightarrow} q_j$
and $q_j \stackrel{t}{\rightarrow} q_j$, where $v, t$ are such that
$\Pi_{\neq i} (v) = \Pi_{\neq i} (t) = (u_2)^p$,
\item
the automaton ${\cal A}'$ admits a path
$q' \stackrel{(u_2)^p}{\rightarrow} q'$,
\item
the set $q'$ contains at least one accepting state of $\cal A$.
\end{itemize}

For those properties, we deduce $(u_2)^p \in U'_{q'}$. It follows that
$((u_2)^p)^{\omega} = (u_2)^{\omega}$ is accepted by the widget associated
to $q'$, hence that ${\cal A}''$ admits the run
\[
q'_0  \stackrel{u_1}{\rightarrow}
q' \stackrel{\varepsilon}{\rightarrow}
q'_{\mathit{repeat}} \stackrel{(u_2)^p}{\rightarrow}
q'_{\mathit{repeat}} \stackrel{(u_2)^p}{\rightarrow}
q'_{\mathit{repeat}} \stackrel{(u_2)^p}{\rightarrow}
\cdots
\]
that accepts $w''$.
\qed
\end{proof}

\subsection{Complexity}
\label{sec:complexity}

The algorithm introduced in Section~\ref{sec:implementation} for
computing an automaton accepting $L_{\forall x \varphi}$
from an automaton $\cal A$ accepting $L_{\varphi}$ runs in time
$2^{O(|\cal A|)}$, where $|\cal A|$ is the size of $\cal A$. This
complexity is tight thanks to the following result.

\begin{theorem}
There exist a family of formulas $\varphi_n$ for $n \in \mathbb{N}$, and
automata ${\cal A}^{(n)}$ accepting the corresponding $L_{\varphi_n}$, such that
$|{\cal A}^{(n)}| = O(n)$, and the number of states of any automaton accepting
$L_{\forall x \varphi_n}$ is at least $O(2^n)$.
\end{theorem}

\begin{proof}
The automaton given in Figure~\ref{fig:automaton-shift} recognizes the
set of models of
\[
\varphi(x_1, x_2, X_1)\,\equiv\, x_2 = x_1 + n \,\wedge\, X_1(x_1)
\Leftrightarrow X_1(x_2),
\]
where $n \geq 1$ is parameter.

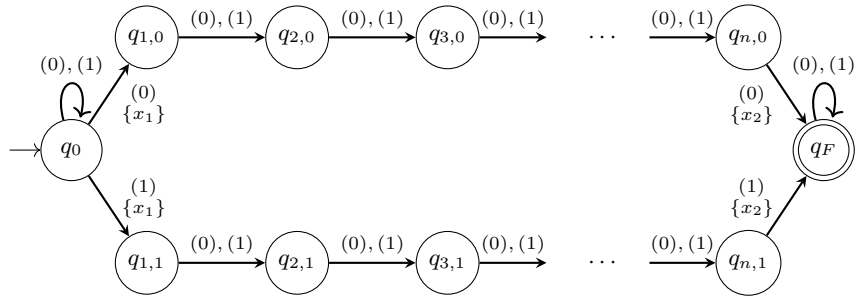
\begin{figure}
\centering

\begin{tikzpicture}
 \node (q0) [state, initial, initial text = {}] at (0,0) {$q_0$};
 \node (q1) [state] at (1,1.5) {$q_{1,0}$};
 \node (q2) [state] at (3,1.5) {$q_{2,0}$};
 \node (q3) [state] at (5,1.5) {$q_{3,0}$};
 \node (q4) at (7,1.5) {~~~~$\cdots$~~~~};
 \node (qn) [state] at (9,1.5) {$q_{n,0}$};
 \node (qf) [state, accepting] at (10,0) {$q_F$};
 \node (q1b) [state] at (1,-1.5) {$q_{1,1}$};
 \node (q2b) [state] at (3,-1.5) {$q_{2,1}$};
 \node (q3b) [state] at (5,-1.5) {$q_{3,1}$};
 \node (q4b) at (7,-1.5) {~~~~$\cdots$~~~~};
 \node (qnb) [state] at (9,-1.5) {$q_{n,1}$};

 \path [-stealth, thick]
     (q0) edge [loop above] node {\scriptsize $(0), (1)$} ()
     (q0) edge [right] node {\scriptsize$\begin{array}{c}\\(0)\\\{ x_1 \}\end{array}$} (q1)
     (q1) edge [above] node {\scriptsize $(0), (1)$} (q2)
     (q2) edge [above] node {\scriptsize $(0), (1)$} (q3)
     (q3) edge [above] node {\scriptsize $(0), (1)$} (q4)
     (q4) edge [above] node {\scriptsize $(0), (1)$} (qn)
     (qn) edge [left] node {\scriptsize$\begin{array}{c}\\(0)\\\{ x_2 \}\end{array}$} (qf)
     (q0) edge [right] node {\scriptsize$\begin{array}{c}(1)\\\{ x_1 \}\\[2ex]\end{array}$} (q1b)
     (q1b) edge [above] node {\scriptsize $(0), (1)$} (q2b)
     (q2b) edge [above] node {\scriptsize $(0), (1)$} (q3b)
     (q3b) edge [above] node {\scriptsize $(0), (1)$} (q4b)
     (q4b) edge [above] node {\scriptsize $(0), (1)$} (qnb)
     (qnb) edge [left] node {\scriptsize$\begin{array}{c}(1)\\\{ x_2 \}\\[2ex]\end{array}$~} (qf)
     (qf) edge [loop above] node {\scriptsize $(0), (1)$} ();

\end{tikzpicture}
\caption{Automaton for $x_2 = x_1 + n \,\wedge\, X_1(x_1)
\Leftrightarrow X_1(x_2)$.}
\label{fig:automaton-shift}
\end{figure}

By removing the variable $x_2$ from all transition labels of this
automaton, one obtains an automaton ${\cal A}^{(n)}$ of size $O(n)$
that recognizes the models of $\exists x_2 (x_2 = x_1 + n \,\wedge\,
X_1(x_1) \Leftrightarrow X_1(x_2))$. One then observes that any
automaton recognizing the models of $\forall x_1 \exists x_2 (x_2 =
x_1 + n \,\wedge\, X_1(x_1) \Leftrightarrow X_1(x_2))$ essentially
checks that the value of the predicate $X_1$ is identical to itself
shifted by $n$ positions, in other words, that it is periodic with the
period $n$.  Such an automaton must therefore have a memory that can
store $n$ bits of information, hence its number of states must at
least be equal to $2^n$.  \qed
\end{proof}

\section{Conclusions}
\label{sec:conclusions}

This paper introduces a method for directly computing the effect of a
first-order universal quantifier on an infinite-word automaton
recognizing the set of models of a formula. It is applicable when the
first-order variables range over the natural numbers and their values
are encoded in unary notation. Among its potential applications, it
provides a solution for deciding the first-order theory
$\langle\mathbb{N}, < \rangle$ extended with uninterpreted unary predicates.

The operation on regular languages that corresponds to the effect of a
universal first-order quantifier has already been studied at the
theoretical level~\cite{Okhotin05}. Our contribution is to provide a
practical algorithm for computing it, that does not require to
complement infinite-word automata. This algorithm has an exponential
worst-case time complexity, which is unavoidable since there exist
automata for which universal quantification incurs an exponential
blowup in their number of states. The main advantage over the
complementation-based approach is however that this exponential cost
is not systematic, since only a fraction of the possible subsets of
states typically need to be constructed. The situation is similar to
the subset construction algorithm for determinizing finite-word
automata, which is able to handle in practice automata with millions
of states, in spite of its worst-case exponential cost.

Our solution is open to many possible improvements, one of them being
to extend the algorithm so as to quantify over several first-order
variables in a single operation. For future work, we plan to
generalize this algorithm to automata over more expressive structures,
such as the automata over linear orders defined in~\cite{BC07}. This
would make it possible to obtain an implementable decision procedure
for, e.g., the first-order theory $\langle\mathbb{R}, < \rangle$ with
uninterpreted unary predicates~\cite{LL66}. Another challenge would be
to develop a similar construction for second-order universal quantification.

\bibliography{paper}

\begin{thebibliography}{10}
\providecommand{\url}[1]{\texttt{#1}}
\providecommand{\urlprefix}{URL }
\providecommand{\doi}[1]{https://doi.org/#1}

\bibitem{BJW05}
Boigelot, B., Jodogne, S., Wolper, P.: An effective decision procedure for
  linear arithmetic over the integers and reals. ACM Tr. Comp. Logic
  \textbf{6}(3),  614--633 (2005)

\bibitem{BL04}
Boigelot, B., Latour, L.: Counting the solutions of {Presburger} equations
  without enumerating them. Theo. Comp. Sc.  \textbf{313}(1),  17--29 (2004)

\bibitem{BC07}
Bruyère, V., Carton, O.: Automata on linear orderings. Journal of Computer and
  System Sciences  \textbf{74}(1),  1--24 (2007)

\bibitem{Buchi60}
Büchi, J.R.: Weak second-order arithmetic and finite automata. Mathematical
  Logic Quarterly  \textbf{6}(1--6),  66--92 (1960)

\bibitem{Buchi62}
Büchi, J.R.: On a decision method in restricted second order arithmetic. In:
  Proc. Intl. Congr. on Logic, Methodology and Philosophy of Science. pp. 1--12
  (1962)

\bibitem{Downey72}
Downey, P.J.: Undecidability of {P}resburger arithmetic with a single monadic
  predicate letter. Tech. rep., Harvard University (1972)

\bibitem{Halpern91}
Halpern, J.Y.: Presburger arithmetic with unary predicates is {$\Pi_1^1$}
  complete. The Journal of Symbolic Logic  \textbf{56}(2),  637--642 (Jun 1991)

\bibitem{Klarlund97}
Klarlund, N.: Mona {\&} {Fido}: The logic-automaton connection in practice. In:
  Proc. 11th CSL Workshop. LNCS, vol.~1414, pp. 311--326. Springer (1997)

\bibitem{LL66}
Läuchli, H., Leonard, J.: On the elementary theory of linear order. Fundamenta
  Mathematicae  \textbf{59}(1),  109--116 (1966)

\bibitem{McNaughton66}
McNaughton, R.: Testing and generating infinite sequences by a finite
  automaton. Information and Control  \textbf{9}(5),  512--530 (1966)

\bibitem{Okhotin05}
Okhotin, A.: The dual of concatenation. Theo. Comp. Sc.  \textbf{345}(2--3),
  425--447 (2005)

\bibitem{Safra88}
Safra, S.: On the complexity of omega-automata. In: Proc. 29th FOCS. pp.
  319--327. {IEEE} Computer Society (1988)

\bibitem{SKR98}
Shiple, T.R., Kukula, J.H., Ranjan, R.K.: A comparison of {Presburger} engines
  for {EFSM} reachability. In: Proc. 10th CAV. LNCS, vol.~1427, pp. 280--292
  (1998)

\bibitem{Speranski13}
Speranski, S.O.: A note on definability in fragments of arithmetic with free
  unary predicates. Archive for Mathematical Logic  \textbf{52}(5-6),  507--516
  (2013)

\bibitem{Thomas90}
Thomas, W.: Automata on infinite objects. In: Handbook of Theoretical Computer
  Science, Volume {B}, pp. 133--191. Elsevier and {MIT} Press (1990)

\bibitem{Vardi07}
Vardi, M.: The {Büchi} complementation saga. In: Proc. 24th STACS. LNCS,
  vol.~4393, pp. 12--22. Springer (2007)

\end{thebibliography}

\end{document}